\documentclass[letterpaper, 10 pt, conference]{ieeeconf}  

\IEEEoverridecommandlockouts                              
\usepackage{graphicx} 
\usepackage{amsmath} 
\usepackage{amssymb}  
\usepackage[ruled]{algorithm2e}
\usepackage{booktabs}
\usepackage{url}
\usepackage{tikz}
\usepackage{siunitx}
\usepackage{stfloats}
\renewcommand{\baselinestretch}{0.965}
\usetikzlibrary{arrows.meta, positioning, decorations.pathreplacing, calc}
\definecolor{nodeblue}{RGB}{70,110,160}
\definecolor{edgegray}{RGB}{180,180,180}
\definecolor{bitgold}{RGB}{190,145,40}
\definecolor{fwblue}{RGB}{55,90,140}
\definecolor{contourgreen}{RGB}{80,130,90}
\definecolor{vertexred}{RGB}{180,60,60}

\usepackage{microtype}
\usepackage[T1]{fontenc}
\newif\ifsubmission
\submissiontrue   

\usepackage{hyperref}
\ifsubmission
  \hypersetup{
    colorlinks  = false,
    pdfborder   = {0 0 0}
  }
\else
  \hypersetup{
    colorlinks  = true,
    linkcolor   = [RGB]{30,80,160},
    citecolor   = [RGB]{30,80,160},
    urlcolor    = [RGB]{30,80,160}
  }
\fi

\hypersetup{
  pdftitle      = {Making Every Bit Count for A-Optimal State Estimation},
  pdfauthor     = {C. Khanpour, D. Turizo, S. Talkington},
  pdfkeywords   = {A-optimal, state estimation, Frank-Wolfe, power systems},
  bookmarksnumbered = false,
  bookmarksopen     = false,
  pdfpagemode       = UseNone,
  pdfstartview      = FitH,
  pdfcreator        = {LaTeX with ieeeconf},
  pdfproducer       = {pdflatex},
  unicode           = true,
  breaklinks        = true
}

\let\ieeeproof\proof
\let\endieeeproof\endproof
\let\proof\relax
\let\endproof\relax
\usepackage{amsthm}
\usepackage{dsfont}
\usepackage{stmaryrd}

\newtheorem{theorem}{Theorem}
\newtheorem{proposition}{Proposition}
\newtheorem{lemma}{Lemma}
\newtheorem{corollary}{Corollary}[proposition]

\theoremstyle{definition}

\theoremstyle{remark}
\newtheorem*{remark}{Remark}

\DeclareSymbolFont{sansops}{OT1}{\sfdefault}{m}{n}
\SetSymbolFont{sansops}{bold}{OT1}{\sfdefault}{b}{n}

\makeatletter
\renewcommand\operator@font{\mathgroup\symsansops}
\makeatother

\DeclareSymbolFont{sfoperators}{OT1}{cmss}{m}{n}
\DeclareSymbolFontAlphabet{\mathsf}{sfoperators}

\makeatletter
\def\operator@font{\mathgroup\symsfoperators}
\makeatother

\newcommand{\R}{\mathbb{R}}

\newcommand{\Z}{\mathbb{Z}}

\newcommand{\E}{\operatorname{ E}}

\renewcommand{\j}{\mathrm{j}}

\newcommand{\vct}[1]{\boldsymbol{#1}}
\newcommand{\mtx}[1]{\boldsymbol{#1}}
\newcommand{\diag}{\operatorname{diag}}

\newcommand{\<}{\langle}
\renewcommand{\>}{\rangle}

\newcommand{\T}{\top}

\newcommand{\trace}{\operatorname{ tr}}

\newcommand{\ip}[2]{\left\<#1, #2\right\>}

\newcommand{\Expec}[2][]{\E_{#1}\left[#2\right]}

\newcommand{\p}[1]{\left(#1\right)}

\newcommand{\set}[1]{\mathcal{#1}}

\DeclareMathOperator*{\argmin}{\operatorname{arg~min}}
\DeclareMathOperator*{\argmax}{\operatorname{arg~max}}

\newcommand{\vb}{\vct{b}}

\newcommand{\vd}{\vct{d}}
\newcommand{\ve}{\vct{e}}

\newcommand{\vg}{\vct{g}}
\newcommand{\vh}{\vct{h}}

\newcommand{\vr}{\vct{r}}
\newcommand{\vs}{\vct{s}}

\newcommand{\vu}{\vct{u}}
\def\vv{\vct{v}}

\newcommand{\vx}{\vct{x}}
\newcommand{\vy}{\vct{y}}
\newcommand{\vz}{\vct{z}}

\newcommand{\vepsilon}{\vct{\epsilon}}

\newcommand{\vkappa}{\vct{\kappa}}

\newcommand{\vmu}{\vct{\mu}}

\newcommand{\vxi}{\vct{\xi}}

\newcommand{\vrho}{\vct{\rho}}

\newcommand{\vzero}{\vct{0}}
\newcommand{\vone}{\vct{1}}

\newcommand{\mA}{\mtx{A}}

\newcommand{\mC}{\mtx{C}}
\newcommand{\mD}{\mtx{D}}

\newcommand{\mH}{\mtx{H}}
\newcommand{\mI}{\mtx{I}}

\newcommand{\mM}{\mtx{M}}

\newcommand{\mP}{\mtx{P}}
\newcommand{\mQ}{\mtx{Q}}

\newcommand{\mS}{\mtx{S}}

\newcommand{\mZ}{\mtx{Z}}
\newcommand{\st}{\operatorname{\sf s.t. }}

\let\proof\ieeeproof
\let\endproof\endieeeproof

\newcommand{\vxhat}{\vct{\hat{x}}}

\newcommand{\floor}[1]{\left\lfloor #1 \right\rfloor}
\renewcommand{\epsilon}{\varepsilon}

\title{\LARGE \bf
Making Every Bit Count for $A$-Optimal State Estimation
}

\author{Cameron Khanpour,\quad Daniel Turizo,\quad Samuel Talkington
\thanks{This material is based upon work supported by the National Science Foundation Graduate Research Fellowship Program under Grant No. DGE-2039655. Any opinions, findings, and conclusions or recommendations expressed in this material are those of the author(s) and do not necessarily reflect the views of the National Science Foundation.}
\thanks{C. Khanpour is with the School of Electrical and Computer Engineering,
        Georgia Institute of Technology, Atlanta, GA 30332, USA.
        {\tt\small ckhanpour3@gatech.edu}}%
\thanks{D. Turizo is with SimpleRose, Inc.,
        St. Louis, MO 63101, USA.
        {\tt\small daniel.turizo@simplerose.com}}%
\thanks{S. Talkington is with the Harvard John A. Paulson School of Engineering and Applied Sciences, Harvard University, Boston, MA 02134, USA, and the Department of Electrical Engineering and Computer Science, University of Michigan, Ann Arbor, MI 48109, USA. 
{\tt\small talks@umich.edu}
}%
}

\begin{document}

\maketitle
\thispagestyle{empty}
\pagestyle{empty}

\begin{abstract}
We study the problem of controlling how a limited communication bandwidth budget is allocated across heterogeneously quantized sensor measurements. The performance criterion is the trace of the error covariance matrix of the linear minimum mean square error (LMMSE) state estimator, i.e., an $A$-optimal design criterion. Minimizing this criterion with a bit budget constraint yields a nonconvex optimization problem. We derive a formula that reduces each evaluation of the gradient to a single Cholesky factorization. This enables efficient optimization by both a projection-free Frank--Wolfe method (with a computable convergence certificate) and an interior point method with L-BFGS Hessian approximation over the problem's continuous relaxation. A largest remainder rounding procedure recovers integer bit allocations with a bound on the quality of the rounded solution. Numerical experiments in IEEE power grid test cases with up to 300 buses compare both solvers and demonstrate that the analytic gradient is the key computational enabler for both~methods. Additionally, the heterogeneous bit allocation is compared to standard uniform bit allocation on the 500 bus IEEE power grid test case. 
\end{abstract}

\section{Introduction}
\label{sec:intro}

Analog-to-digital \textit{quantization} methods improve computational and communication efficiency by mapping continuous measurements to discrete intervals \cite{gray_quantization_survey_1998}. While coarse quantization accelerates computations \cite{gholami2022survey}, it introduces \textit{nonlinear}, typically \textit{non-Gaussian}, measurement noise \cite{plan_generalized_2016,thrampoulidis_quantized_lasso_2020}, which contrasts with the additive Gaussian noise assumptions that are frequently used for analyzing state estimation algorithms.

At the same time, sensing technologies are increasingly widespread in modern engineering practice. For example, advanced metering infrastructure (AMI) has become widespread in electric power distribution infrastructure, where \textit{variable precision} sensing has emerged as a promising operational paradigm for communication-constrained sensing. To address these challenges, we cast heterogeneous quantizer bit allocation under a global bandwidth budget as an $A$-optimal design problem for linear minimum mean square error (LMMSE) estimation\textemdash i.e., minimizing the trace of the covariance matrix of the error of the LMMSE estimator. The resource allocation acts on measurement precisions and induces a nonconvex optimization problem that falls outside the scope of existing sensor selection methods.

\paragraph{Related Work}
In \cite{joshi2009sensor}, the discrete sensor \textit{placement} problem was studied under a $D$-optimal criterion, whose relaxed problem is solved without guarantees of optimality to the discrete problem. Budget-constrained $D$-optimal design was recently studied in~\cite{wang_algorithms_2025}. Similarly, \cite{krause_efficient_2008} applied sensor placement for water networks; they focus on a class of objective functions that possess other structured properties, such as submodularity (this property is useful to prove theoretical guarantees of greedy based algorithms). While $A$-optimal design does not have (super)modularity in general, recent research in \cite{chamon_approximate_2017} showed that $A$-optimal design enjoys \textit{approximate} supermodularity. Combinatorial approximation algorithms for $A$-optimal design have also been developed~\cite{madan_combinatorial_2019,nikolov_aoptimal_2019}. Frank--Wolfe algorithms~\cite{frank1956fw} for the classical $A$-optimal design problem were analyzed in~\cite{ahipasaoglu_first-order_2015}, where the information matrix is linear in the design weights and the problem is convex. The exponential mapping between sensor bitrate and precision in our formulation breaks this linearity and introduces nonconvexity, requiring special convergence analysis. The work of~\cite{arce_convexity_depth_2020} develops a convex reformulation of a related problem under an MMSE constraint; an earlier paper from the same group~\cite{arce_utility_2016} proposes a greedy bit assignment heuristic. This line of inquiry enables valuable insights on the energy savings obtained by a chosen bit allocation policy.

Another component of this work is the incorporation of the effects of quantization and communication bandwidth in state estimation. Quantization has a significant effect on the accuracy of state estimation in scenarios limited by bandwidth. In communication-constrained settings, such as remote monitoring over low-bandwidth links or large-scale sensor networks with shared communication channels, the number of bits allocated to each measurement directly determines the quantization noise variance and, consequently, the estimation quality. This connects to the classical theory of optimal experimental design~\cite{fedorov1972theory}, where measurement resources are allocated to minimize estimation error. 

Quantized estimation has been addressed in the Kalman filtering~\cite{sshuli2007QKF,emsechu2008DQKF}, sensor network placement~\cite{kekatos2012optimal}, and topology learning~\cite{talkington_quantized_2026} settings. In power systems, the literature on state estimation has also focused on the importance of heterogeneous data sources, communication constraints, and measurement quality for large scale networks \cite{cheng_survey_2023}.

\paragraph{Contributions}
In this paper, we study optimal bit allocation for LMMSE state estimation under a global communication budget. We formulate the problem as a nonconvex $A$-optimal design problem with the following contributions:
\begin{enumerate}

\item We present two algorithms: i.) A first-order Frank--Wolfe algorithm with a closed form linear minimization oracle (Proposition~\ref{prop:lmo}), and show that the minimum Frank--Wolfe gap over iterates converges to zero at rate $O(1/\sqrt{T})$ iterations (Theorem~\ref{thm:fw-main}), providing a computable convergence certificate. ii.) A second-order interior point algorithm. We derive a closed form gradient for both FW and the interior point method to reduce each evaluation to a single Cholesky factorization of the information matrix. Further, the interior point method is accelerated with a L-BFGS Hessian approximation and converges in very few iterations. 
\item We introduce an independent largest remainder rounding procedure (Algorithm~\ref{alg:largest-remainder-rounding}) that maps continuous relaxations to integer allocations that are guaranteed to be feasible and have bounded solution quality gap.
\item We evaluate the speed and solution quality of the algorithms in the practical setting of experimental design for power network state estimation within different regimes (``sensor rich'' and limited bandwidth). We also validate the problem formulation of heterogeneous bit allocation to a standard uniform bit allocation, achieving up to $53\%$ improvement in the limited bandwidth regime. 
\end{enumerate}

The remainder of the paper is organized as follows. Section~\ref{sec:prob} introduces the measurement model, the LMMSE estimator, and the resulting bit allocation problem under a global communication budget. Section~\ref{sec:alg} presents the proposed solution methods for the relaxed problem, including the analytic gradient, the Frank--Wolfe algorithm and its convergence guarantee, and an interior point approach. Section~\ref{sec:rounding} describes the rounding procedure used to recover integer bit allocations from relaxed solutions along with guarantees. Section~\ref{sec:numerical} reports numerical experiments on IEEE power system test cases, comparing solver performance and evaluating the benefits of heterogeneous allocation relative to uniform allocation. Finally, Section~\ref{sec:conclusion} concludes the paper and outlines directions for future work.

\section{Problem Formulation}
\label{sec:prob}
Figure~\ref{fig:system_diagram} outlines the setting of our problem on quantized network state estimation. We formulate the measurement model and allocation problem mathematically as follows.
\subsection{Measurement Model and Estimation}
\label{sec:meas}

Given a fixed sensing matrix~$\mH \in \R^{m \times d}$ whose rows are~$\{\vh_i^\top\}_{i=1}^m$, where each~$\vh_i \in \R^d$, corresponding to~$m$ linear projections of an unknown state vector~$\vx \in \R^d$, we consider heterogeneously quantized measurements with quantization bin widths~$\{\Delta_i\}_{i=1}^m$. We assume throughout that $\vh_i \neq \vzero$ for all~$i$, i.e., every sensor observes a nontrivial linear combination of the state.

Each quantized measurement~$i$ is generated from a uniformly dithered quantization function~$\set{Q}_i : \R \to \R$ (see~\cite{gray_quantization_survey_1998,thrampoulidis_quantized_lasso_2020}) with bin width $\Delta_i >0$ such that
\begin{equation}
    \label{eq:unif_quant_measmt}
    \set{Q}_i(\ip{\vh_i}{\vx}) = \Delta_i \cdot\p{\floor{\frac{\ip{\vh_i}{\vx} + \tau_i}{\Delta_i}} + \frac{1}{2}} - \tau_i,
\end{equation}
where $\tau_i \sim \mathrm{Uniform}(-\frac{\Delta_i}{2},\frac{\Delta_i}{2})$, and $\Expec[\tau_i]{\set{Q}_i(\ip{\vh_i}{\vx})\,|\,\vh_i}=\ip{\vh_i}{\vx}$, i.e., $\set{Q}_i(\ip{\vh_i}{\vx})$ is conditionally unbiased. The independent dither $\tau_i$ is noise purposely applied prior to quantization; see~\cite{gray_quantization_survey_1998, gray_dithered_quantizer_1993}.
The quantization error is well approximated by the additive model
\begin{equation}
\label{eq:additive_model}
\vy := \mH\vx + \vz,
\end{equation}
where $\Expec[]{\vz \vz^\T} = \frac{1}{12}\diag(\Delta_1^2,\ldots,\Delta_m^2)$. Set $d_i := \Delta_i^2/12$ and $\mD := \diag(d_1,\ldots,d_m)$. Suppose $\vx$ is zero-mean with prior covariance $\mC_{\vx} \succ 0$. The LMMSE state estimator is
\begin{equation}
    \label{eq:lmmse-state-est}
    \vxhat = \mC_{\vx} \mH^\T\p{\mH\mC_{\vx}\mH^\T + \mD}^{-1}\!\vy,
\end{equation}
with error covariance
\begin{equation}
\label{eq:error-cov}
\mC_{\vepsilon} = \p{\mH^\T\mD^{-1}\mH + \mC_{\vx}^{-1}}^{-1}\!.
\end{equation}
Quantization error is generally signal-dependent, so an additive noise model can be inaccurate, especially at coarse resolutions \cite{gray_quantization_survey_1998}. Our use of uniformly dithered quantizers in~\eqref{eq:unif_quant_measmt} is intended  to mitigate this dependence and justify the covariance model \cite{gray_dithered_quantizer_1993}. Accordingly, the optimization problem below should be interpreted as the $A$-optimal design problem induced by this dithered additive noise LMMSE model.

\begin{figure}[t]
\centering
\begin{tikzpicture}[
    sensor/.style={circle, draw=nodeblue, fill=nodeblue!20,
        minimum size=13pt, inner sep=0pt, font=\scriptsize,
        text=nodeblue!80!black},
    bit/.style={fill=bitgold, draw=bitgold!50!black, line width=0.15pt}
]
\node[sensor] (n1) at (0.3, 1.6) {1};
\node[sensor] (n2) at (1.6, 2.2) {2};
\node[sensor] (n3) at (3.0, 1.7) {3};
\node[sensor] (n4) at (0.5, 0.2) {4};
\node[sensor] (n5) at (2.0, 0.4) {5};
\node[sensor] (n6) at (3.3, 0.7) {6};
\foreach \i/\j in {1/2, 2/3, 1/4, 2/4, 2/5, 3/5, 3/6, 5/6} {
    \draw[edgegray, thick] (n\i) -- (n\j);
}
\foreach \n/\b in {1/4, 2/7, 3/2, 4/1, 5/6, 6/3} {
    \foreach \k in {1,...,\b} {
        \pgfmathsetmacro\ys{(\k-1)*0.11}
        \fill[bit] ([xshift=3pt, yshift=\ys cm]n\n.east)
            rectangle ++(0.08cm, 0.08cm);
    }
    \pgfmathsetmacro\labely{(\b)*0.1 - 0.02}
    \node[font=\tiny, text=bitgold!80!black, anchor=south west]
        at ([xshift=2pt, yshift=\labely cm]n\n.east) {\b};
}
\draw[edgegray!70, rounded corners=2pt, thin]
    (4.15, 0.35) rectangle (6.55, 2.45);
\node[font=\footnotesize, text=edgegray!70!black, anchor=north west, scale=0.8, yshift=-4pt]
    at (4.25, 2.4) {$b_3 = 2$ bits};
\draw[edgegray!60, dashed, thin, ->]
    ([xshift=1pt]n3.east) -- (4.15, 1.4);
\begin{scope}[shift={(4.35, 0.55)}, xscale=1.15, yscale=1.4]
    \foreach \y in {0.125, 0.375, 0.625, 0.875} {
        \draw[edgegray!30, very thin, dashed] (0, \y) -- (1.8, \y);
    }
    \draw[nodeblue, semithick, opacity=0.5]
        plot[smooth, domain=0:1.8, samples=60]
        (\x, {0.35*sin(\x*200+30) + 0.5});
    \draw[bitgold, semithick]
        (0, 0.625) -- (0.08, 0.625) -- (0.08, 0.875) --
        (0.5, 0.875) -- (0.5, 0.625) -- (0.75, 0.625) --
        (0.75, 0.375) -- (1.0, 0.375) -- (1.0, 0.125) --
        (1.4, 0.125) -- (1.4, 0.375) -- (1.65, 0.375) --
        (1.65, 0.625) -- (1.8, 0.625);
    \node[font=\footnotesize, text=bitgold!80!black, anchor=north, scale=0.8, yshift=-4pt]
        at (0.9, -0.05) {$2^{b_3}\!=2^{2} = \!4$ levels};
\end{scope}
\node[anchor=north, font=\small] at (3.3, -0.3) {
    $\vy = \mH\vx + \vz, \quad
    \textcolor{bitgold}{\textstyle\sum_{i=1}^{m} b_i} \le B$
};
\end{tikzpicture}
\caption{Bit allocation on a sensor network. Each node~$i$ receives
$b_i$ bits (\,$\textcolor{bitgold}{\blacksquare}$\,=\,1\,bit) subject
to budget $\textcolor{bitgold}{\sum_i b_i} \le B$.
Inset: 2-bit quantization maps an analog measurement to
$2^{b_i}\!=\!4$ discrete levels.}
\label{fig:system_diagram}
\end{figure}

\subsection{Bandwidth Allocation Problem}
\label{sec:bandwidth}

We allocate a limited number of quantization bits across $m$ channels to minimize the MSE. For a $b_i$-bit uniform quantizer with dynamic range $R_i$,
\begin{equation}
\label{eq:delta-d-rho-b}
\Delta_i=\frac{R_i}{2^{b_i}},
\quad
d_i=\frac{R_i^2}{12\cdot 4^{b_i}},
\quad
\rho_i:=d_i^{-1}
=
\kappa_i\,4^{b_i},
\end{equation}
where $\kappa_i:=12/R_i^2$. Allocating more bits to channel $i$ increases its precision exponentially. The integer program is
\begin{equation}
\label{eq:bit-alloc-integer}
\min_{\vb\in\Z_+^m}
\quad
\trace\bigl(\mC_{\vepsilon}(\vb)\bigr)
\quad
\st
\quad
\vone^\T \vb \leq B,
\end{equation}
where $B\in \Z_+$. A continuous relaxation replaces $\Z_+^m$ by $\R_+^m$. Define $\mathcal{B} := \{\vb\in\R_+^m:\ \vone^\T \vb \leq B\}$ and write
\begin{equation}
\label{eq:M-and-Ceps-b}
\mM(\vb)
:=
\mC_{\vx}^{-1}
+
\mH^\T\diag(\vrho(\vb))\mH,
\quad
\mC_{\vepsilon}(\vb)
:=
\mM(\vb)^{-1}.
\end{equation}
The relaxed problem is
\begin{equation}
\label{eq:bit-alloc-relaxed}
\min_{\vb\in\mathcal{B}}
\quad
F(\vb)
:=
\trace\bigl(\mC_{\vepsilon}(\vb)\bigr).
\end{equation}

The feasible set $\mathcal{B}$ is compact and convex, but $F$ is generally nonconvex in the bit variables due to the exponential reparameterization $\rho_i = \kappa_i 4^{b_i}$. Introducing an auxiliary function $f(\vrho) := \trace((\mC_{\vx}^{-1} + \mH^\T\diag(\vrho)\mH)^{-1})$ in precision variables, we have $F(\vb)=f(\vrho(\vb))$. While $f$ is convex in $\vrho$, the chain rule yields
\begin{equation}
\begin{split}
\label{eq:hessian-chain-rule}
\nabla^2 F(\vb)
=
(\ln 4)^2
&\Bigl[
\diag(\vrho)\nabla^2 f(\vrho)\diag(\vrho) \\
&+
\diag\bigl(\vrho\odot \nabla f(\vrho)\bigr)
\Bigr],
\end{split}
\end{equation}
where the second term is negative semidefinite (since $\nabla f \leq \vzero$ componentwise), so convexity does not transfer from $\vrho$ to $\vb$.

Although $F(\vb)$ is nonconvex in the bit variables, the function
$f(\vrho) = \trace\bigl((\mC_{\vx}^{-1} + \mH^\T\diag(\vrho)\mH)^{-1}\bigr)$
is convex in the precision variables~$\vrho$.
However, the budget constraint $\vone^\T \vb \le B$ becomes
$\sum_{i=1}^m \log_4(\rho_i/\kappa_i) \le B$ in $\vrho$-space, which is
\emph{not} convex (it is the sublevel set of a concave function).
Thus, the $\vb$-space formulation trades a nonconvex objective for a convex
(polyhedral) feasible set, which is precisely what enables the Frank--Wolfe
method with a closed form linear minimization oracle.
This tradeoff is complementary to the reformulation
of~\cite{arce_convexity_depth_2020}, which achieves a convex objective
at the cost of a more complex constraint geometry.

\begin{lemma}[Gradient with respect to bits]
\label{lemma:grad-b}
For each measurement $i=1,\dots,m$,
\begin{equation}
\label{eq:grad-b}
\frac{\partial F}{\partial b_i}(\vb)
=
-(\ln 4)\,\rho_i(\vb)\,
\vh_i^\T \mC_{\vepsilon}(\vb)^2 \vh_i.
\end{equation}
Equivalently,
\begin{equation}
\label{eq:grad-b-vector}
\nabla F(\vb)
=
-(\ln 4)\,
\vrho(\vb)\odot
\left[\vh_i^\T \mC_{\vepsilon}(\vb)^2 \vh_i\right]_{i=1}^m.
\end{equation}
In particular, $\partial F/\partial b_i < 0$ whenever $\vh_i \neq \vzero$.
Each gradient evaluation reduces to a single Cholesky factorization of~$\mM(\vb)$. The diagonal entries $\vh_i^\T \mC_{\vepsilon}^2 \vh_i$ are extracted column-by-column from $\mC_{\vepsilon}\mH^\T$ without forming the full $m\times m$ product, at cost $O(d^2 m)$ dominated by the $O(d^3)$ factorization when $m = O(d)$.
\end{lemma}

\begin{proof}
Define the following
\[
\mA(\vrho)
:=
\mC_{\vx}^{-1}
+
\mH^\T\diag(\vrho)\mH,
\quad
\mC_{\vepsilon}(\vrho)=\mA(\vrho)^{-1}.
\]
Since
$\partial \mA/\partial \rho_i
=
\mH^\T \ve_i\ve_i^\T \mH
=
\vh_i\vh_i^\T$,
the derivative identity for matrix inverses gives
\[
\frac{\partial \mC_{\vepsilon}}{\partial \rho_i}
=
-
\mC_{\vepsilon}
\frac{\partial \mA}{\partial \rho_i}
\mC_{\vepsilon}
=
-\mC_{\vepsilon}\vh_i\vh_i^\T\mC_{\vepsilon}.
\]
Taking traces,
\begin{align*}
\frac{\partial f}{\partial \rho_i}
&=
\trace\p{\frac{\partial \mC_{\vepsilon}}{\partial \rho_i}}
=
-\trace\bigl(\mC_{\vepsilon}\vh_i\vh_i^\T \mC_{\vepsilon}\bigr)
=
-\vh_i^\T \mC_{\vepsilon}^2 \vh_i.
\end{align*}
Since $\mC_{\vepsilon}\succ 0$, the quadratic form $\vh_i^\T \mC_{\vepsilon}^2 \vh_i$ is nonnegative. Now $F(\vb)=f(\vrho(\vb))$ with $\rho_i(\vb)=\kappa_i4^{b_i}$, so by chain rule
\[
\frac{\partial F}{\partial b_i}
=
\frac{\partial f}{\partial \rho_i}\cdot
\frac{\partial \rho_i}{\partial b_i} = \left(-\vh_i^\T \mC_{\vepsilon}(\vb)^2 \vh_i \right) \cdot (\ln 4)\rho_i(\vb),
\]
which yields~\eqref{eq:grad-b}. The vector form~\eqref{eq:grad-b-vector} follows by stacking. If $\vh_i\neq \vzero$, then $\vh_i^\T \mC_{\vepsilon}(\vb)^2\vh_i>0$ because $\mC_{\vepsilon}(\vb)\succ 0$, so the derivative is strictly negative.
\end{proof}

\begin{proposition}[Budget saturation]
\label{prop:budget-saturation}
If $\vh_i\neq \vzero$ for all $i$ and $\vb^\star$ minimizes~\eqref{eq:bit-alloc-relaxed}, then $\vone^\T \vb^\star = B$.
\end{proposition}

\begin{proof}
Consider any feasible $\vb$ such that $\vone^\T \vb < B$. Then $\vb + t\ve_i\in\mathcal{B}$ for some index $i$ and all sufficiently small $t>0$. By Lemma~\ref{lemma:grad-b},
$\partial F/\partial b_i(\vb)<0$,
so increasing $b_i$ decreases the objective, which means that $\vb$ is not a minimum. Conversely, if $\vb^\star$ minimizes~\eqref{eq:bit-alloc-relaxed}, we must have that $\vone^\T \vb^\star = B$.
\end{proof}

\begin{proposition}[Lipschitz continuity of the gradient]
\label{prop:lipschitz-grad}
The function $F$ is $C^\infty$ on $\R^m$, and $\nabla F$ is $L$-Lipschitz on $\mathcal{B}$ with
\begin{equation}
\label{eq:lipschitz-constant}
L
=
\frac{(\ln 4)^2}{4}\|\mC_{\vx}\|_2.
\end{equation}
\end{proposition}

\begin{proof}
Recall that $F(\vb)=f(\vrho(\vb))$ and $\rho_i(\vb)=\kappa_i4^{b_i}$, where
\[
f(\vrho)
=
\trace\Bigl(
\bigl(
\mC_{\vx}^{-1}
+
\mH^\T\diag(\vrho)\mH
\bigr)^{-1}
\Bigr).
\]
Since each coordinate map $\vb\mapsto \rho_i(\vb)$ is smooth and strictly positive on $\R^m$, and matrix inversion is smooth on the positive definite cone, $F$ is $C^\infty$ on $\R^m$. For a fixed $\vb\in\R^m$ and by Lemma~\ref{lemma:grad-b},
\(
\partial f/\partial\rho_i=-\vh_i^\T\mC_{\vepsilon}^2\vh_i.
\)
Differentiating once more, using
\(
\partial\mC_{\vepsilon}/\partial\rho_j
=-\mC_{\vepsilon}\vh_j\vh_j^\T\mC_{\vepsilon},
\)
gives
\begin{equation}
\label{eq:hess-rho-proof}
\frac{\partial^2 f}{\partial\rho_j\partial\rho_i}
=
2\,(\vh_i^\T\mC_{\vepsilon}\vh_j)
(\vh_i^\T\mC_{\vepsilon}^2\vh_j).
\end{equation}
Set $\mS:=\diag(\sqrt{\rho_1},\ldots,\sqrt{\rho_m})$ and define
\[
\mP:=\mS\mH\mC_{\vepsilon}\mH^\T\mS,
\qquad
\mQ:=\mS\mH\mC_{\vepsilon}^2\mH^\T\mS.
\]
Both matrices are positive semidefinite. The matrix inversion lemma gives that
\[
\mI-\mP
=
\bigl(\mI+\mS\mH\mC_{\vx}\mH^\T\mS\bigr)^{-1}
\succ0,
\]
so $0\preceq\mP\preceq\mI$. Combining~\eqref{eq:hessian-chain-rule} and~\eqref{eq:hess-rho-proof},
\begin{equation}
\label{eq:hess-b-proof}
\nabla^2 F(\vb)
=
(\ln4)^2\bigl[2(\mP\odot\mQ)
-\diag(Q_{11},\ldots,Q_{mm})\bigr],
\end{equation}
where $\odot$ denotes the entrywise product. By the Schur product theorem, $\mP\odot\mQ$ and $(\mI-\mP)\odot\mQ$ are positive semidefinite. Therefore
\[
0\preceq\mP\odot\mQ
\preceq\mI\odot\mQ
=\diag(Q_{11},\ldots,Q_{mm}).
\]
It follows from~\eqref{eq:hess-b-proof} that, for every $\vv\in\R^m$,
\begin{align*}
\bigl|\vv^\T\nabla^2F(\vb)\vv\bigr|
&\leq (\ln4)^2\sum_{i=1}^m Q_{ii}v_i^2\\
&\leq (\ln4)^2\max_i Q_{ii}\,\|\vv\|_2^2.
\end{align*}
Since the Hessian is symmetric, this gives
\begin{equation}
\label{eq:hess-norm-bound}
\|\nabla^2F(\vb)\|_2
\leq (\ln4)^2\max_i Q_{ii}.
\end{equation}
Fix $i$, and write
\(
\mA_i:=\mC_{\vx}^{-1}
+\sum_{k\neq i}\rho_k\vh_k\vh_k^\T\succ0
\)
and $s_i:=\vh_i^\T\mA_i^{-1}\vh_i\geq0$. By the Sherman--Morrison--Woodbury identity,
\[
\mC_{\vepsilon}\vh_i
=
\frac{\mA_i^{-1}\vh_i}{1+\rho_i s_i}.
\]
Since $\mA_i\succeq\mC_{\vx}^{-1}$, we have $\mA_i^{-1}\preceq\mC_{\vx}$. The spectral theorem then gives
\[
\mA_i^{-2}
\preceq \|\mA_i^{-1}\|_2\,\mA_i^{-1}
\preceq \|\mC_{\vx}\|_2\,\mA_i^{-1}.
\]
Consequently,
\begin{align*}
Q_{ii}
&=\rho_i\|\mC_{\vepsilon}\vh_i\|_2^2
=\frac{\rho_i\vh_i^\T\mA_i^{-2}\vh_i}{(1+\rho_i s_i)^2}\\
&\leq \|\mC_{\vx}\|_2\,
\frac{\rho_i s_i}{(1+\rho_i s_i)^2}
\leq \frac{\|\mC_{\vx}\|_2}{4},
\end{align*}
where the last inequality uses $(1+t)^2\geq4t$ for $t\geq0$. Substituting into~\eqref{eq:hess-norm-bound} yields
\(
\|\nabla^2F(\vb)\|_2\leq L
\)
for all $\vb\in\R^m$. Since $\mathcal{B}$ is convex, integrating the Hessian along the segment between any $\vb,\vb'\in\mathcal{B}$ gives
\[
\|\nabla F(\vb')-\nabla F(\vb)\|_2
\leq L\,\|\vb'-\vb\|_2.
\]
Thus $\nabla F$ is $L$-Lipschitz continuous on $\mathcal{B}$.
\end{proof}

\section{Bit Allocation Procedure}
\label{sec:alg}

We solve the relaxed problem~\eqref{eq:bit-alloc-relaxed} using the Frank--Wolfe (FW) method~\cite{frank1956fw}; see~\cite{jaggi_frankwolfe_2013,pokutta_frank-wolfe_2024} for introductions. Unlike the convex setting of~\cite{ahipasaoglu_first-order_2015}, our objective is nonconvex in~$\vb$, so we apply the nonconvex FW analysis of~\cite{lacoste-julien_fw_nonconvex_2016}. FW is well suited because the feasible set is compact and convex, the objective is smooth, and the linear minimization oracle (LMO) admits a closed form solution. Figure~\ref{fig:fw-simplex}
illustrates this relationship.
\subsection{Gradient and Linear Minimization Oracle}
\label{sec:lmo}

Given a feasible iterate $\vb\in\mathcal{B}$, the FW linear minimization oracle solves
\begin{equation}
\label{eq:lmo-def}
\vs(\vb)
\in
\argmin_{\vs\in\mathcal{B}}
\ip{\nabla F(\vb)}{\vs}.
\end{equation}
The FW direction is $\vd(\vb):=\vs(\vb)-\vb$, and the FW gap is
\begin{equation}
\label{eq:fw-gap-def}
g_{\mathrm{FW}}(\vb)
:=
\max_{\vs\in\mathcal{B}}
\ip{\vb-\vs}{\nabla F(\vb)}
=
-\ip{\vd(\vb)}{\nabla F(\vb)}.
\end{equation}
As usual, $g_{\mathrm{FW}}(\vb)\ge 0$, with equality if and only if $\vb$ is a first-order stationary point.

\begin{proposition}[Closed form LMO]
\label{prop:lmo}
Let $\vg:=\nabla F(\vb)\in\R^m$. Then an optimal solution of~\eqref{eq:lmo-def} is
\begin{equation}
\label{eq:lmo-solution}
\vs^\star
=
\begin{cases}
B\,\ve_{i^\star},
& \text{if }\min_i g_i < 0
\text{ and } i^\star\in\argmin_i g_i,\\
\vzero,
& \text{if } g_i\ge 0 \text{ for all } i.
\end{cases}
\end{equation}
Under the standing assumption $\vh_i\neq \vzero$ for all $i$ (so $\vg<\vzero$ by Lemma~\ref{lemma:grad-b}), this simplifies to
\begin{equation}
\label{eq:lmo-solution-negative-grad}
\vs^\star = B\,\ve_{i^\star},
\quad
i^\star\in\argmin_{1\le i\le m} [\nabla F(\vb)]_i.
\end{equation}
\end{proposition}

\begin{proof}
The feasible set $\mathcal{B}$ is a polytope with extreme points $\vzero, B\ve_1,\dots,B\ve_m$. Since $\vs\mapsto \ip{\vg}{\vs}$ is linear, an optimum is attained at a vertex. Evaluating gives $\ip{\vg}{\vzero}=0$ and $\ip{\vg}{B\ve_i}=Bg_i$, so the minimum is $\min\{0, B\min_i g_i\}$, yielding~\eqref{eq:lmo-solution}. Simplification~\eqref{eq:lmo-solution-negative-grad} follows from Lemma~\ref{lemma:grad-b}.
\end{proof}

\begin{corollary}[Explicit Frank--Wolfe gap]
\label{cor:fw-gap-explicit}
Under the assumption $\vh_i\neq \vzero$ for all $i$,
\[
g_{\mathrm{FW}}(\vb)
=
\ip{\vb}{\nabla F(\vb)}
-
B\min_{1\le i\le m} [\nabla F(\vb)]_i.
\]
\end{corollary}

\begin{proof}
By definition of FW gap in~\eqref{eq:fw-gap-def}, $g_{\mathrm{FW}}(\vb) = \ip{\vb-\vs(\vb)}{\nabla F(\vb)}$ where $\vs(\vb)\in\argmin_{\vs\in\mathcal{B}}\ip{\nabla F(\vb)}{\vs}$. Since $\nabla F(\vb)<\vzero$ componentwise by Lemma~\ref{lemma:grad-b}, Proposition~\ref{prop:lmo} gives $\vs(\vb)=B\ve_{i^\star}$ with $i^\star\in\argmin_i[\nabla F(\vb)]_i$. Hence
$g_{\mathrm{FW}}(\vb) = \ip{\vb}{\nabla F(\vb)} - B\min_i [\nabla F(\vb)]_i$.
\end{proof}

\subsection{Convergence Guarantee}
\label{sec:convergence}

Starting from $\vb^{(0)}\in\mathcal{B}$, the FW method generates
\[
\vs^{(t)}
\in
\argmin_{\vs\in\mathcal{B}}
\ip{\nabla F(\vb^{(t)})}{\vs},
\quad
\vd^{(t)}:=\vs^{(t)}-\vb^{(t)},
\]
and updates $\vb^{(t+1)} = (1-\gamma_t)\vb^{(t)}+\gamma_t \vs^{(t)}$ with step size
\begin{equation}
\label{eq:fw-update}
\gamma_t = \min\left\{
\frac{g_{\mathrm{FW}}(\vb^{(t)})}{2LB^2},\,1
\right\}.
\end{equation}
Since $\mathcal{B}$ is convex, every iterate remains feasible. The step size~\eqref{eq:fw-update} uses $\mathrm{diam}(\mathcal{B})=\sqrt{2}B$ by property of simplices.

\begin{figure}[t]
\centering
\begin{tikzpicture}[>=Stealth,
    iterate/.style={circle, fill=fwblue, inner sep=1.8pt},
    lmovert/.style={circle, fill=vertexred, inner sep=2.5pt}]
\coordinate (v1) at (0, 0);
\coordinate (v2) at (3.4, 0);
\coordinate (v3) at (1.7, 2.94);
\draw[nodeblue, thick] (v1) -- (v2) -- (v3) -- cycle;
\node[below left, font=\small] at (v1) {$B\ve_1$};
\node[below right, font=\small] at (v2) {$B\ve_2$};
\node[above, font=\small] at (v3) {$B\ve_3$};
\begin{scope}
    \clip (v1) -- (v2) -- (v3) -- cycle;
    \draw[contourgreen, thin, opacity=0.45]
        plot[smooth cycle, tension=3] coordinates {
        (2.1, 0.7) (2.4, 0.9) (2.2, 1.15) (1.85, 1.0)};
    \draw[contourgreen, thin, opacity=0.45]
        plot[smooth cycle, tension=1.2] coordinates {
        (2.55, 0.4) (2.9, 0.75) (2.6, 1.3) (1.9, 1.45)
        (1.35, 1.1) (1.4, 0.55)};
    \draw[contourgreen, thin, opacity=0.45]
        plot[smooth cycle, tension=1.5] coordinates {
        (2.85, 0.2) (3.2, 0.65) (2.9, 1.5) (2.15, 1.85)
        (1.1, 1.6) (0.55, 0.95) (0.65, 0.4) (1.5, 0.12)};
\end{scope}
\node[font=\scriptsize, text=contourgreen!70!black] at (2.4, 0.2) {$F(\vb)$};
\node[iterate] (bt) at (1.4, 0.55) {};
\node[font=\small, text=fwblue, anchor=south] at ([xshift=13pt, yshift=-9pt]bt.west) {$\vb^{(t)}$};
\draw[->, thick, bitgold] (bt) -- ++(0.4, 0.5)
    node[below right=-3pt, font=\small, text=bitgold!80!black] {$-\nabla F$};
\node[lmovert] at (v3) {};
\node[above, font=\small, text=vertexred, yshift=-8pt, xshift=26pt]
    at (v3) {$\vs^{(t)}\!=\!B\ve_{i^\star}$};
\draw[fwblue, dashed, thin] (bt) -- (v3);
\node[iterate] (bt1) at ($(bt)!0.3!(v3)$) {};
\node[font=\small, text=fwblue, above=0pt, xshift=15pt, yshift=-3pt] at (bt1) {$\vb^{(t+1)}$};
\draw[decorate, decoration={brace, amplitude=2.5pt, raise=4pt},
    thin, fwblue!50]
    (bt.center) -- (bt1.center)
    node[midway, below=-7pt, xshift=-12pt, font=\scriptsize, text=fwblue!80!black] {$\gamma_t$};
\node[anchor=north, font=\small] at (1.7, -0.4) {
    $g_{\mathrm{FW}}(\vb^{(t)})
    = \ip{\vb^{(t)} - \vs^{(t)}}{\nabla F(\vb^{(t)})}$
};
\end{tikzpicture}
\caption{Frank--Wolfe iteration on the budget simplex $\mathcal{B}$ for $m=3$.
The LMO selects vertex $\vs^{(t)}$; the next iterate lies on
$[\vb^{(t)}, \vs^{(t)}]$ with step $\gamma_t$.}
\label{fig:fw-simplex}
\end{figure}
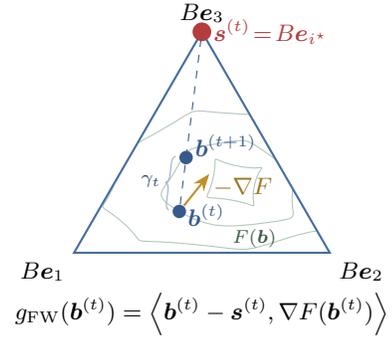

\begin{theorem}[FW convergence for relaxed bit allocation \cite{lacoste-julien_fw_nonconvex_2016}]
\label{thm:fw-main}
Assume $\vh_i\neq \vzero$ for all $i$, and let $\{\vb^{(t)}\}_{t\ge 0}$ be generated by~\eqref{eq:fw-update} starting from any $\vb^{(0)}\in\mathcal{B}$. Define $h_0 := F(\vb^{(0)})-\min_{\vb\in\mathcal{B}} F(\vb)$. Then
\begin{equation}
\label{eq:fw-gap-rate}
\min_{0\le t\le T}
g_{\mathrm{FW}}(\vb^{(t)})
\le
\frac{\max\{2h_0,\; 2LB^2\}}{\sqrt{T+1}}.
\end{equation}
To achieve $g_{\mathrm{FW}}(\vb^{(t)})\le \varepsilon$, it suffices to perform $T+1 \ge \max\{2h_0,2LB^2\}^2/\varepsilon^2$ iterations.
\end{theorem}

\begin{proof}
By Proposition~\ref{prop:lipschitz-grad}, $F$ is continuously differentiable on $\R^m$ with $L$-Lipschitz gradient on the compact convex set $\mathcal{B}$. The FW curvature constant $C_F$ over $\mathcal{B}$ satisfies
\[
C_F
\le
L\,\mathrm{diam}(\mathcal{B})^2
=
2LB^2,
\]
where $\mathrm{diam}(\mathcal{B})=\sqrt{2}B$ follows from the vertex structure. Applying~\cite[Theorem~1]{lacoste-julien_fw_nonconvex_2016} with $C:=2LB^2$ gives
\[
\min_{0\le t\le T}
g_{\mathrm{FW}}(\vb^{(t)})
\le
\frac{\max\{2h_0,\; C\}}{\sqrt{T+1}}
=
\frac{\max\{2h_0,\; 2LB^2\}}{\sqrt{T+1}},
\]
which is~\eqref{eq:fw-gap-rate}. Rearranging yields the iteration complexity.
\end{proof}

Theorem~\ref{thm:fw-main} guarantees convergence to a first-order stationary point at rate $O(\varepsilon^{-2})$. The stationarity measure is the computable FW gap, obtained at no extra cost once the LMO is solved. 
In practice, the global bound~$L$ in the iteration complexity can be conservative.
One may replace the step size in Algorithm~\ref{alg:fw-bit-allocation}
with an adaptive Lipschitz search: initialize $\hat{L}_0=1$;
at each iteration, halve $\hat{L}$, then double until
$F(\vb+\gamma_t\vd^{(t)})\leq F(\vb^{(t)})-\gamma_t g_{\mathrm{FW}}(\vb^{(t)})/2$,
capping at~$L$. Since $\hat{L}_t\leq L$,
Theorem~\ref{thm:fw-main} still applies.

\begin{algorithm}[t]
\small
\DontPrintSemicolon
\caption{Frank--Wolfe for relaxed bit allocation}
\label{alg:fw-bit-allocation}
\KwIn{$\mH$, $\mC_{\vx}$, $\vkappa$, $B$, $L$, $\varepsilon$, $T_{\max}$}
\KwOut{feasible $\vb^{(t)}\in\mathcal{B}$}

$\vb^{(0)}\gets (B/m)\vone$\;
\For{$t=0,1,\dots,T_{\max}$}{
    $\vrho^{(t)} \gets \vkappa \odot 4^{\vb^{(t)}}$\;
    $\mM^{(t)} \gets \mC_{\vx}^{-1} + \mH^\T \diag(\vrho^{(t)}) \mH$\;
    Compute a Cholesky factor of $\mM^{(t)}$\;
    Solve $\mM^{(t)}\mZ^{(t)}=\mH^\T$ using this factor\;
    $[\vg^{(t)}]_i \gets -(\ln 4)\,\rho_i^{(t)}\,
    \|\mZ^{(t)}_{:,i}\|_2^2$, $i=1,\ldots,m$\;

    $i_t \gets \argmin_{1\le i\le m} [\vg^{(t)}]_i$\;
    $\vs^{(t)} \gets B\,\ve_{i_t}$\;

    $g_t \gets \ip{\vb^{(t)}-\vs^{(t)}}{\vg^{(t)}}$\;

    \lIf{$g_t \le \varepsilon$}{
        \Return{$\vb^{(t)}$}
    }

    $\gamma_t \gets \min\left\{\dfrac{g_t}{2LB^2},\,1\right\}$\;
    $\vb^{(t+1)} \gets (1-\gamma_t)\vb^{(t)} + \gamma_t \vs^{(t)}$\;
}
\Return{$\vb^{(T_{\max}+1)}$}\;
\end{algorithm}

\subsection{Interior Point Method with Analytic Gradient}
\label{sec:ip-lbfgs}

The gradient formula from Lemma~\ref{lemma:grad-b} can also be supplied to general purpose nonlinear programming solvers. We use the interior point solver Ipopt~\cite{wachter_implementation_2006} via JuMP~\cite{DunningHuchetteLubin2017} with the option \texttt{hessian\_approximation = ``limited-memory''}, which activates an L-BFGS Hessian approximation. Ipopt replaces the constrained problem~\eqref{eq:bit-alloc-relaxed} with a sequence of log-barrier subproblems
\[
\min_{\vb}\; F(\vb) - \mu \sum_{i=1}^m \ln b_i - \mu \ln \bigl(B - \vone^\T \vb\bigr),
\]
where $\mu > 0$ is driven to zero. Each subproblem is solved by a damped Newton method on the primal-dual KKT system. Since the constraint structure is simple ($m$ bound constraints plus one linear budget constraint), the KKT system is inexpensive to solve once the gradient is available.

The Hessian $\nabla^2 F(\vb)$ is costly to compute exactly. Instead, we use the limited-memory BFGS (L-BFGS) approximation, which maintains a low rank estimate from recent gradient differences. With L-BFGS, each iteration requires only the gradient $\nabla F(\vb)$ from Lemma~\ref{lemma:grad-b}, the same $O(d^3)$ Cholesky factorization used by Frank--Wolfe, plus $O(m)$ work for the L-BFGS update and KKT solve. 

\begin{remark}
\label{rem:per-iter-cost}
Both Frank--Wolfe and the interior point method are bottlenecked by the same $O(d^3)$ Cholesky factorization for the gradient. FW adds an $O(m)$ LMO step, whereas Ipopt adds an $O(m)$ KKT solve. The key difference is iteration count; while FW and interior point methods both globally converge sublinearly at rate $O(1/\sqrt{T})$, IPMs also exhibit local superlinear convergence, typically requiring 20--50 iterations. FW provides a computable convergence certificate (the FW gap), while Ipopt reports KKT residuals.
\end{remark}

\section{Rounding Procedure}
\label{sec:rounding}

The relaxed problem~\eqref{eq:bit-alloc-relaxed} returns a continuous solution
$\bar{\vb}\in\mathcal{B}$, whereas the original problem~\eqref{eq:bit-alloc-integer}
requires $\vb\in\Z_+^m$. We compute
\[
\vr:=\bar{\vb}-\lfloor \bar{\vb}\rfloor \in [0,1)^m,
\quad
R_{\mathrm{rem}}:=\vone^\T \vr,
\]
where $\lfloor\cdot\rfloor$ is computed componentwise.
Since $\vone^\T\bar{\vb}=B$ by Proposition~\ref{prop:budget-saturation} and $\vone^\T\lfloor \bar{\vb}\rfloor\in\Z_+$,
the residual budget is an integer satisfying
$0\le R_{\mathrm{rem}}\le m-1$.

Similarly to the Quota Method for apportionment \cite{balinski_quota_1975} and the rounding procedures for approximate experimental designs~\cite{pukelsheim_efficient_1992}, we use \textit{largest remainder rounding} to round up the $R_{\mathrm{rem}}$ largest
components of $\vr$ and round the rest down. Equivalently,  
\begin{equation}
\label{eq:largest-remainder-def}
\vxi
\in
\argmax_{\substack{\vxi\in\{0,1\}^m\\ \vone^\T \vxi = R_{\mathrm{rem}}}}
\ip{\vr}{\vxi},
\quad
\hat{\vb}
:=
\lfloor \bar{\vb}\rfloor + \vxi.
\end{equation}

\begin{algorithm}[t]
\small
\DontPrintSemicolon
\caption{Largest remainder rounding procedure}
\label{alg:largest-remainder-rounding}
\KwIn{continuous solution $\bar{\vb}\in\mathcal{B}$ with $\vone^\T\bar{\vb}=B$}
\KwOut{integral $\hat{\vb}\in\Z_+^m$ with $\vone^\T\hat{\vb}=B$}

$\vb^{(0)} \gets \lfloor \bar{\vb}\rfloor$\;
$\vr \gets \bar{\vb}-\lfloor \bar{\vb}\rfloor$\;
$R_{\mathrm{rem}} \gets B-\vone^\T\vb^{(0)}$\;
$\hat{\vb} \gets \vb^{(0)}$\;

find a set $S \subseteq \{1,\dots,m\}$ with $|S|=R_{\mathrm{rem}}$ such that
\[
S \in \argmax_{\substack{T\subseteq\{1,\dots,m\}\\ |T|=R_{\mathrm{rem}}}}
\sum_{i\in T} r_i
\]
(i.e., $S$ indexes the $R_{\mathrm{rem}}$ largest components of $\vr$)\;

\ForEach{$i\in S$}{
    $\hat b_i \gets \hat b_i + 1$\;
}
\Return{$\hat{\vb}$}\;
\end{algorithm}

\begin{proposition}[Feasibility and nearest point property of largest remainder rounding]
\label{prop:largest-remainder-distance}
Let $\bar{\vb}\in\mathcal{B}$ satisfy $\vone^\T \bar{\vb}=B$, and define
$\vr=\bar{\vb}-\lfloor \bar{\vb}\rfloor$ and $R_{\mathrm{rem}}=\vone^\T \vr$.
Then $\hat{\vb}$ defined by~\eqref{eq:largest-remainder-def}
belongs to $\Z_+^m$, satisfies $\vone^\T \hat{\vb}=B$, and solves
\begin{equation}
\label{eq:largest-remainder-nearest}
\hat{\vb}
\in
\argmin_{\substack{\hat{\vb}\in\Z_+^m,\\
\hat{\vb}=\lfloor \bar{\vb}\rfloor+\vxi}}
\|\hat{\vb}-\bar{\vb}\|_2^2 \ \text{\rm s.t. } \vxi\in\{0,1\}^m,\ 
\vone^\T\vxi=R_{\mathrm{rem}}.
\end{equation}
Moreover,
\begin{equation}
\label{eq:largest-remainder-distance-bound}
\|\hat{\vb}-\bar{\vb}\|_2^2
\le
\sum_{i=1}^m r_i(1-r_i).
\end{equation}
\end{proposition}

\begin{proof}
By construction, $\vxi\in\{0,1\}^m$ and
$\vone^\T \vxi=R_{\mathrm{rem}}$, so
$\hat{\vb}=\lfloor \bar{\vb}\rfloor+\vxi\in\Z_+^m$
and
\[
\vone^\T \hat{\vb}
=
\vone^\T \lfloor \bar{\vb}\rfloor + \vone^\T \vxi
=
(B-R_{\mathrm{rem}})+R_{\mathrm{rem}}
=
B.
\]

Now let $\hat{\vb}=\lfloor \bar{\vb}\rfloor+\vxi$ with
$\vxi\in\{0,1\}^m$ and $\vone^\T \vxi=R_{\mathrm{rem}}$. Since
$\bar{\vb}=\lfloor \bar{\vb}\rfloor+\vr$, we obtain
\[
\|\hat{\vb}-\bar{\vb}\|_2^2
=
\|\vxi-\vr\|_2^2
=
\|\vr\|_2^2 + \|\vxi\|_2^2 - 2\ip{\vr}{\vxi}.
\]
Because $\vxi\in\{0,1\}^m$ with $\vone^\T\vxi=R_{\mathrm{rem}}$, we have
$\|\vxi\|_2^2 = R_{\mathrm{rem}}$. Hence
\[
\|\hat{\vb}-\bar{\vb}\|_2^2
=
\|\vr\|_2^2 + R_{\mathrm{rem}} - 2\ip{\vr}{\vxi}.
\]
Therefore minimizing $\|\hat{\vb}-\bar{\vb}\|_2^2$ over all feasible
$\vxi$ is equivalent to maximizing $\ip{\vr}{\vxi}$, which proves~\eqref{eq:largest-remainder-nearest}.

For~\eqref{eq:largest-remainder-distance-bound}, note that
$\vr\in[0,1]^m$ and $\vone^\T\vr=R_{\mathrm{rem}}$, so $\vr$ is feasible for the
continuous relaxation of the maximization problem in~\eqref{eq:largest-remainder-def}. The feasible polytope $\{\vu\in[0,1]^m:\vone^\T\vu=R_{\mathrm{rem}}\}$ has only integral vertices. Indeed, two fractional coordinates admit opposite small perturbations preserving feasibility, so such a point cannot be a vertex. Exactly one fractional coordinate is impossible because $R_{\mathrm{rem}}$ is an integer. A linear objective attains its maximum at a vertex, so the relaxation and~\eqref{eq:largest-remainder-def} have the same optimal value. Hence
\(
\ip{\vr}{\vxi}
\ge
\ip{\vr}{\vr}
=
\|\vr\|_2^2.
\)
Substituting into the norm identity above gives
\begin{align*}
\|\hat{\vb}-\bar{\vb}\|_2^2
&=
\|\vr\|_2^2 + R_{\mathrm{rem}} - 2\ip{\vr}{\vxi}\\
&\le
R_{\mathrm{rem}}-\|\vr\|_2^2
=
\sum_{i=1}^m r_i(1-r_i),
\end{align*}
as claimed.
\end{proof}

\begin{theorem}[Rounding gap for largest remainder rounding]
\label{thm:largest-remainder-gap}
Assume $\vh_i\neq \vzero$ for all $i$, and let $\bar{\vb}\in\mathcal{B}$
satisfy $\vone^\T\bar{\vb}=B$. Suppose $\bar{\vb}$ is a KKT point of the relaxed problem~\eqref{eq:bit-alloc-relaxed}. Then the largest remainder rounding
$\hat{\vb}$ defined by~\eqref{eq:largest-remainder-def} satisfies
\begin{equation}
\label{eq:largest-remainder-gap}
F(\hat{\vb})
\le
F(\bar{\vb})
+
\frac{L}{2}
\sum_{i=1}^m r_i(1-r_i).
\end{equation}
In particular,
\begin{equation}
\label{eq:largest-remainder-gap-simplified}
F(\hat{\vb})
\le
F(\bar{\vb})
+
\frac{L}{2} \min\left\{R_{\mathrm{rem}}, m/4\right\}.
\end{equation}
\end{theorem}

\begin{proof}
Since $\bar{\vb},\hat{\vb}\in\mathcal{B}$ and $\mathcal{B}$ is convex,
$L$-smoothness of $F$ on $\mathcal{B}$ gives the following quadratic upper bound
\begin{equation}
\label{eq:largest-remainder-descent}
F(\hat{\vb})
\le
F(\bar{\vb})
+
\ip{\nabla F(\bar{\vb})}{\hat{\vb}-\bar{\vb}}
+
\frac{L}{2}\|\hat{\vb}-\bar{\vb}\|_2^2.
\end{equation}
Let
\(
\boldsymbol{\delta}
:=
\hat{\vb}-\bar{\vb}
=
\vxi-\vr.
\)
Since $\vone^\T \vxi = \vone^\T \vr = R_{\mathrm{rem}}$, we have
$\vone^\T \boldsymbol{\delta}=0$.
Let $(\lambda,\vmu)$ be KKT multipliers for the constraints
$\vone^\T \vb \le B$ and $\vb\ge \vzero$, so
\[
\nabla F(\bar{\vb}) + \lambda \vone - \vmu = \vzero,
\quad
\lambda \ge 0,
\quad
\vmu \ge \vzero,
\quad
\mu_i \bar b_i = 0 \quad \forall i.
\]
If $r_i>0$, then $\bar b_i>0$, so complementary slackness gives $\mu_i=0$.
If $r_i=0$, then $\delta_i=\xi_i-r_i=0$ because largest remainder
rounding only rounds up coordinates with positive fractional part. Therefore
\begin{align*}
\ip{\nabla F(\bar{\vb})}{\boldsymbol{\delta}}
&=
\sum_{i=1}^m \left(-\lambda+\mu_i\right)\delta_i \\
&=
-\lambda \sum_{i=1}^m \delta_i =
-\lambda\,\vone^\T\boldsymbol{\delta}
=
0.
\end{align*}
Substituting into~\eqref{eq:largest-remainder-descent} and applying
Proposition~\ref{prop:largest-remainder-distance} yields
\[
F(\hat{\vb})
\le
F(\bar{\vb})
+
\frac{L}{2}\|\hat{\vb}-\bar{\vb}\|_2^2
\le
F(\bar{\vb})
+
\frac{L}{2}\sum_{i=1}^m r_i(1-r_i),
\]
which proves~\eqref{eq:largest-remainder-gap}. The simplified bound
follows from $r_i(1-r_i)\le r_i$ and $r_i(1-r_i)\le 1/4$ for each $i$.
\end{proof}

\begin{remark}
\label{rem:largest-remainder-fw}
The bound in Theorem~\ref{thm:largest-remainder-gap} requires KKT stationarity of
the relaxed solution, which is satisfied by interior point solutions.
Frank--Wolfe only guarantees a small FW gap, so the bound applies approximately. In practice, the bound depending on the global Lipschitz constant can be loose. Some future work would be to tighten this bound using local Lipschitz constant or exploiting specific structural properties of a given problem's sensing matrix such as sparsity. 
\end{remark}


\begin{figure*}[!tb]
    \centering
    \includegraphics[width=0.98\textwidth]{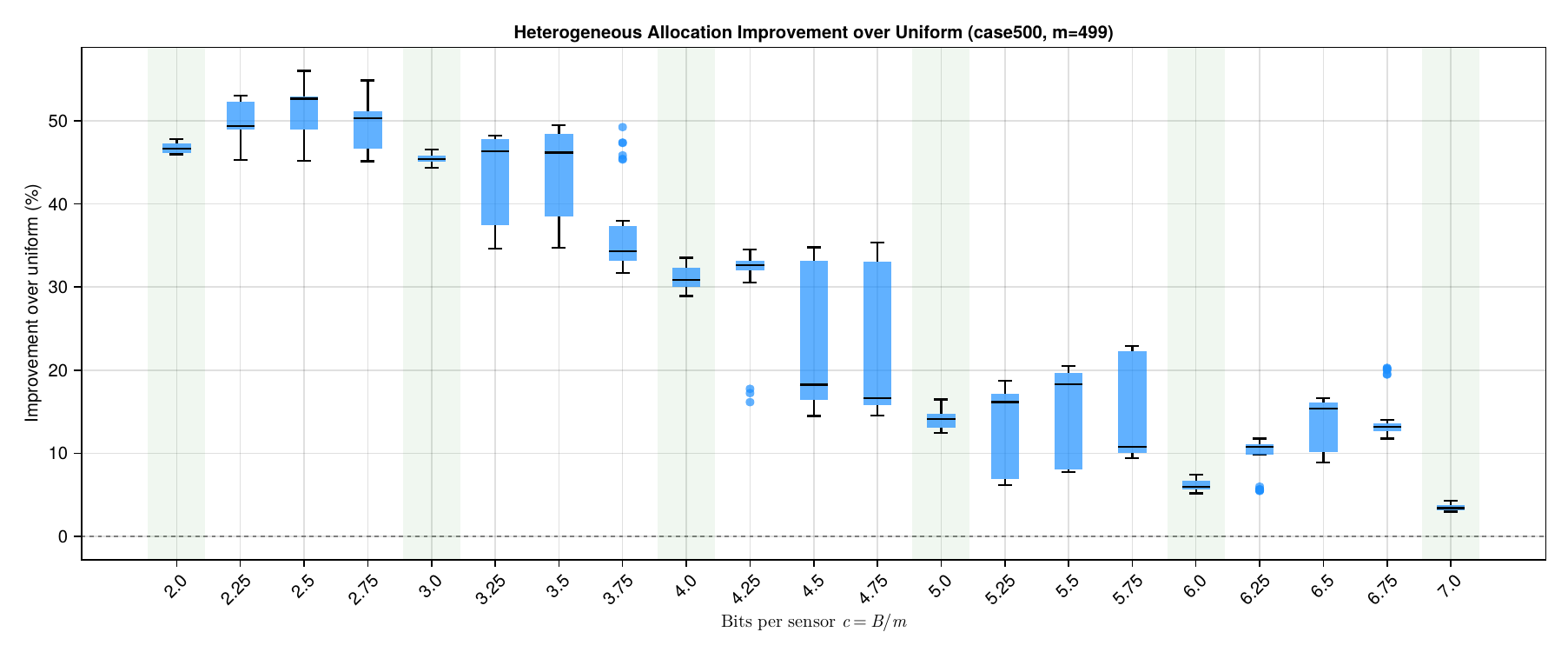}
    \caption{Percentage improvement of the optimized heterogeneous bit allocation~\eqref{eq:bit-alloc-relaxed} over uniform allocation~$\vb = \vone\cdot\lfloor B/m \rfloor$ on the \texttt{case500} test system ($m=499$ sensors). Each boxplot shows 30 randomized instances at the given budget level~$c = B/m$. Shaded columns denote integer values of~$c$.}
    \label{fig:heterogeneous-vs-uniform-case500}
\end{figure*}

\section{Numerical Experiments}
\label{sec:numerical}

We evaluate the proposed Frank--Wolfe algorithm (Algorithm~\ref{alg:fw-bit-allocation}) and compare it against an interior point method (Ipopt) with analytic gradient and L-BFGS Hessian approximation. Both methods use the gradient computation from Lemma~\ref{lemma:grad-b}. All experiments were conducted on an AMD Ryzen 5 7600X3D (6 cores, 4.1\,GHz) with 48\,GB RAM using Julia~1.12.
\subsection{Comparison of Solver Computation Time}

\paragraph{Solver Configuration}
The proposed method uses the FW algorithm with the short step rule~\eqref{eq:fw-update}, which computes the step size as
$\gamma_t = \min\bigl\{ g_t / (2L B^2),\; 1 \bigr\}$
using the analytic Lipschitz constant $L = (\ln 4)^2\|\mC_{\vx}\|_2/4$. The solver runs for at most 500 iterations with a duality gap tolerance of $10^{-6}$ and a 600s wall clock time limit. Each iteration requires one Cholesky factorization of the $d \times d$ information matrix $\mM(\vb)$ via LAPACK, and the information matrix is assembled using a BLAS symmetric rank-$k$ update. 

The interior point baseline uses Ipopt~\cite{wachter_implementation_2006} via JuMP~\cite{DunningHuchetteLubin2017} with the analytic gradient from Lemma~\ref{lemma:grad-b} and L-BFGS Hessian approximation, initialized at $\vb = (B/m) \mathbf{1}$ with a 600s time limit. Each Ipopt iteration requires the same $O(d^3)$ Cholesky factorization as FW. After the continuous solve, both methods apply the same largest remainder rounding procedure from Algorithm~\ref{alg:largest-remainder-rounding} to obtain integer bit allocations.

\paragraph{Test Cases}
We use 11 power grid test cases from the PGLib-OPF benchmark library~\cite{babaeinejadsarookolaee2021}. The sensing matrix $\mH$ is the grounded bus susceptance matrix from DC power flow with the slack bus removed to ensure positive definiteness. State dimensions range from $d = 13$ (case14) to $d = 299$ (case300), with $m = d$ measurements corresponding to power injections at all non slack buses. For each test case, we generate 30 independent problem instances by randomizing the channel precision constants $\kappa_i \sim \mathrm{Uniform}(0.8, 1.2)$, $i = 1, \ldots, m$. The bit budget is set to $B = 2m$ (2 bits per sensor). The prior covariance is $\mC_{\vx} = \mI$.

Table~\ref{tab:solve_times} summarizes solve times across the 30 randomized instances for each test case. Both methods are bottlenecked by one LAPACK Cholesky factorization of the information matrix $\mM(\vb)$ per iteration, with the gradient extracted via $[\mH \mC_{\vepsilon}^2 \mH^\top]_{ii}$ without forming the full $m \times m$ product. The FW LMO adds only $O(m)$ work per iteration via~\eqref{eq:lmo-solution}; the interior point KKT solve is similarly $O(m)$ with L-BFGS. The speed difference is determined by iteration count; the interior point method converges in 20--50 iterations, while FW often used most of the max 500 iterations. Across all instances where both solvers converged, FW and Ipopt+$\nabla$ returned similar objective values up to the solver tolerances. Table~\ref{tab:solve_times} therefore isolates runtime differences rather than solution quality differences. %

Table~\ref{tab:rounding_gap} reports the median rounding gap $F(\hat{\vb}) - F(\bar{\vb})$ and the theoretical bound from Theorem~\ref{thm:largest-remainder-gap} across six IEEE test cases, each evaluated over 30 random instances with $\kappa_i \sim \mathrm{Uniform}(0.8, 1.2)$. In particular, the ratio of the actual gap to the bound ranges from $10^{-4}$ (case14) to $10^{-6}$ (case240, case300), indicating that the bound is highly conservative in practice. This conservatism may stem from $L$ being a global worst case Lipschitz constant, whereas a \textit{local} Lipschitz constant could give a tighter bound.

\begin{table}[tb]
    \centering
    \caption{Solve time comparison: Stock Ipopt, FW, and Ipopt $+$ analytic gradient. Mean $\pm$ std.\ dev.\ over 30 trials, 10 min. limit.}
    \label{tab:solve_times}
    \vspace{-4pt}
    \resizebox{\columnwidth}{!}{
    \begin{tabular}{lrrrr}
        \toprule
        \textbf{Case} & \textbf{$m$} & \textbf{Ipopt (s)} & \textbf{FW (s)} & \textbf{Ipopt+$\nabla$ (s)} \\
        \midrule
        case14  & 13  & $5.03  \pm 1.02$ & $0.098 \pm 0.0023$ & $0.033 \pm 0.14$ \\
        case30  & 29  & $2.73  \pm 0.16$ & $0.24 \pm 0.0093$ & $0.014 \pm 0.0010$ \\
        case57  & 56  & $7.52  \pm 0.31$ & $0.89 \pm 0.030$ & $0.022 \pm 0.0020$ \\
        case73  & 72  & $12.1  \pm 0.97$ & $1.76 \pm 0.090$ & $0.026 \pm 0.0040$ \\
        case118 & 117 & $121.2  \pm 12.6$ & $7.64 \pm 0.22$ & $0.052 \pm 0.0040$ \\
        case162 & 161 & $252.9  \pm 26.1$ & $27.44 \pm 1.1$ & $0.070 \pm 0.0060$ \\
        case179 & 178 & $442.7  \pm 80.3$ & $41.05 \pm 1.6$ & $0.16 \pm 0.032$ \\
        case197 & 196 & $>600$ (timeout) & $50.67 \pm 1.5$ & $0.16 \pm 0.020$ \\
        case200 & 199 & $>600$ (timeout) & $52.59 \pm 2.7$ & $0.12 \pm 0.028$ \\
        case240 & 239 & $>600$ (timeout) & $77.82 \pm 2.4$ & $0.11 \pm 0.056$ \\
        case300 & 299 & $>600$ (timeout) & $207.46 \pm 8.6$ & $0.43 \pm 0.062$ \\
        \bottomrule
    \end{tabular}}
\end{table}

\begin{table}[tb]
    \centering
    \caption{Median experimental rounding quality vs theoretical bound over 30 random instances and $B = 2m$ total bit budget.}
    \label{tab:rounding_gap}
    \vspace{-4pt}
    \resizebox{\columnwidth}{!}{
    \begin{tabular}{lcccc}
        \toprule
        \textbf{Case} & \textbf{$m$} & \textbf{$F(\hat{\vb})-F(\bar{\vb})$}
            & \textbf{Rounding bound~\eqref{eq:largest-remainder-gap}} & \textbf{Median ratio} \\
        \midrule
        case14 & 13 & $1.92\times10^{-2}$ & $6.73\times10^{-1}$ & $2.83\times10^{-2}$ \\
        case30 & 29 & $2.48\times10^{-2}$ & $1.07\times10^{0}$ & $2.36\times10^{-2}$ \\
        case57 & 56 & $4.13\times10^{-2}$ & $2.22\times10^{0}$ & $1.83\times10^{-2}$ \\
        case200 & 199 & $5.76\times10^{-2}$ & $8.37\times10^{0}$ & $6.83\times10^{-3}$ \\
        case240 & 239 & $1.80\times10^{-2}$ & $8.20\times10^{0}$ & $2.20\times10^{-3}$ \\
        case300 & 299 & $8.38\times10^{-2}$ & $1.14\times10^{1}$ & $7.32\times10^{-3}$ \\
        \bottomrule
    \end{tabular}}
\end{table}


\begin{figure}[ht!]
    \centering
    \includegraphics[width=0.98\linewidth]{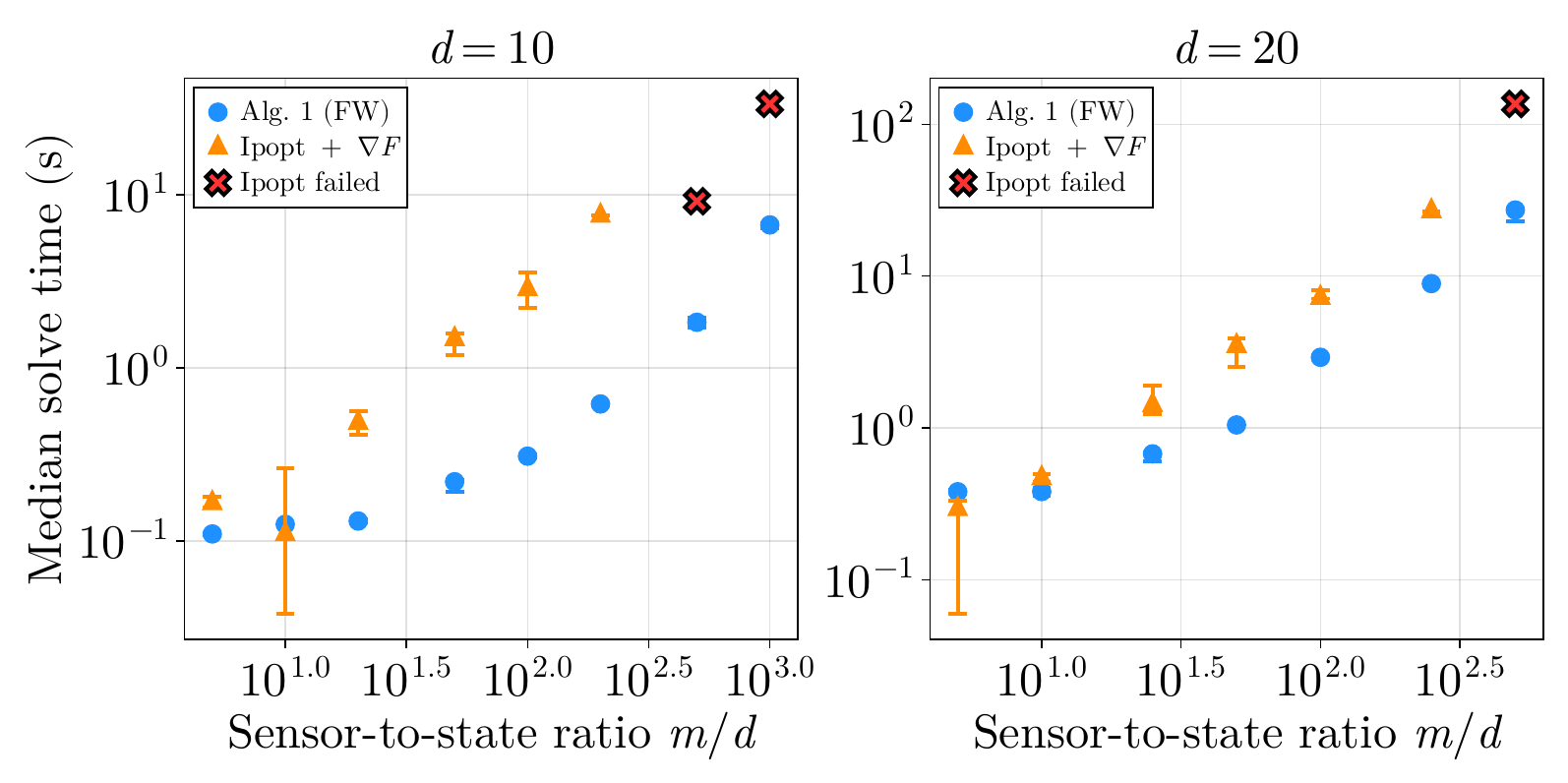}
    \caption{Median solve time vs. sensor-to-state ratio $m/d$ for
    $d \in \{10, 20\}$ with $B = 2d$.  Error bars show the
    IQR over 30 instances. Frank--Wolfe
    scales linearly in $m$ (for $m \gg d$) with $O(d^3 + d^2 m)$ per-iteration
    cost. Ipopt fails to converge entirely beyond
    $m/d \approx 200$ ($d{=}10$) and $100$ ($d{=}20$).}
    \label{fig:sensor-scaling}
\end{figure}

In Figure \ref{fig:sensor-scaling}, we analyzed a ``sensor rich'' ($m \gg d$) regime in which FW outperforms Ipopt with gradient acceleration in computation time. We generate random Gaussian instances with $\mH \in \mathbb{R}^{m \times d}$, $\kappa_i \sim \mathcal{U}(0.8, 1.2)$, and bit budget $B = 2d$. For each state dimension $d \in \{10, 20\}$, we sweep the sensor-to-state ratio $m/d$ from 5 to 1000 and solve 30 independent instances per configuration using both Algorithm~1 and Ipopt with analytic gradients and L-BFGS. Median solve times and ranges are reported with red \texttimes{} markers indicate configurations where Ipopt failed to converge on all 30 instances.

\subsection{Comparison of Problem Formulations}
\label{sec:heterogeneous-value}

We evaluate the benefit of solving~\eqref{eq:bit-alloc-relaxed} relative to the uniform baseline~$\vb = \vone \cdot \lfloor B/m \rfloor$ on the \texttt{case500} power system test case ($m = d=499$) across per-sensor budgets~$c = B/m \in [2, 7]$. For each budget level, 30 instances are generated with quantization gains~$\kappa_i \sim \mathrm{Uniform}(0.8, 1.2)$, and the continuous relaxation is solved with Ipopt+$\nabla$ method followed by largest remainder rounding.

As shown in Figure~\ref{fig:heterogeneous-vs-uniform-case500}, heterogeneous allocation yields median improvements of approximately~$47\%$ at~$c = 2$, rising to~$53\%$ at~$c = 2.5$ before settling near~$50\%$ and~$45\%$ at~$c = 2.75$ and~$c = 3$, respectively, then declining to~$31\%$ at~$c = 4$ and~$14\%$ at~$c = 5$, and falling to~$3.4\%$ at~$c = 7$ as abundant access to bandwidth levels the playing field. The gains at integer~$c$ remain large (e.g., $45\%$ at~$c = 3$, $31\%$ at~$c = 4$), confirming that the dominant source of improvement is the optimizer's ability to concentrate precision on high information sensors rather than an artifact of baseline rounding. These results demonstrate that heterogeneous allocation is most valuable in the bandwidth-constrained regime that motivates this work.

\section{Conclusions}
\label{sec:conclusion}

We proposed an optimal bit allocation method for state estimation with variable precision measurements. The method reduces the MSE of a state estimator by controlling how a limited number of bits are allocated to sensors.

\paragraph{Discussion}
A first-order and second-order method was proposed to solve the corresponding nonconvex optimization problem. The algorithms demonstrated the ability to heterogeneously allocate a limited bit budget across quantized measurements, significantly reducing the $A$-optimal design criterion for the state estimator. We derived a closed form gradient formula that reduces computing the gradient to a single Cholesky factorization. We presented two methods that exploit this gradient: a Frank--Wolfe method with a closed form LMO and guaranteed $O(1/\sqrt{T})$ convergence rate guarantee in terms of the FW gap, and an interior point method with L-BFGS Hessian approximation. Numerical experiments on IEEE power grid test cases show that the interior point method converges in fewer iterations, as expected for a second-order method. Additionally, the Frank--Wolfe method, being a first order method compared to a second order interior point method, is more memory efficient, allowing it to scale to large problems where storing the L-BFGS Hessian itself becomes infeasible. 

\paragraph{Future work}
This work opens a wide array of future directions. This includes adding stochastic trace estimation for scaling to larger problem sizes, such as the classical Hutchinson estimator or more modern approaches~\cite{meyer_hutch_nodate}. This paper has thus far considered only an offline state estimation problem. Future work should extend these methods to dynamic bit allocation in time-varying systems. Another promising direction is the consideration of a broader class of objectives that measure communication cost and performance criteria complementary to the $A$-optimal design criterion. Future work will consider general classes of convex communication cost functionals, and applying decision-focused learning (DFL) principles to optimize over the dynamic ranges of the channels. 


\bibliographystyle{IEEEtran}
\bibliography{cdcrefs}

\end{document}